\documentclass[conference]{IEEEtran}
\usepackage{amsmath,amssymb,amsfonts}
\usepackage{caption}
\usepackage{algorithmic}
\usepackage{graphicx}
\usepackage{subcaption}
\usepackage{float}
\usepackage{cite}
\usepackage{textcomp}
\usepackage{xcolor}
\usepackage{amsthm}
\IEEEoverridecommandlockouts
\newtheorem{proposition}{Proposition}

\newcommand{\be}{\begin{equation}}
\newcommand{\ee}{\end{equation}}
\newcommand{\bea}{\begin{eqnarray}}
\newcommand{\eea}{\end{eqnarray}}

\begin{document}

\title{Making AFDM Secure Against Eavesdroppers: A Phase Function Design Approach\\

\thanks{This publication has emanated from research conducted with the financial support of Research Ireland under the US-Ireland R\&D Partnership Programme Grant Numbers 24/US/4013 and 21/US/3757 and CONNECT- grant 13/RC/2077 P2.}
}

\author{\IEEEauthorblockN{Hengxuan Liu\IEEEauthorrefmark{1}, Vincent Savaux\IEEEauthorrefmark{2}
and 
Arman Farhang\IEEEauthorrefmark{1}}
 
\IEEEauthorblockA{\IEEEauthorrefmark{1}Department of Electronic and Electrical Engineering, Trinity College Dublin, Ireland \\}
\IEEEauthorblockA{\IEEEauthorrefmark{2}b\textless{}\textgreater{}com, Cesson-Sévigné, France \\}
 
 \IEEEauthorblockA{
Email: \{liuh6, arman.farhang\}@tcd.ie, vincent.savaux@b-com.com}

}
\maketitle

\begin{abstract}
Affine frequency division multiplexing (AFDM) has recently emerged as a promising waveform for high-mobility communications due to its resilience to Doppler effects and its advantages for integrated sensing and communication (ISAC). 
AFDM modulates transmit data symbols using chirp subcarriers with two adjustable parameters. One is used for dealing with the Doppler effect and the second parameter can be used for physical layer security (PLS). In this paper, we focus on designing the second chirp parameter in the form of a generic phase function to enhance the robustness of the waveform against brute-force demodulation by the eavesdropper. In particular, we first derive a design criterion that reveals the brute-force demodulation complexity depends on the first derivative of the phase function. Then, we introduce a family of phase functions that can increase the brute-force demodulation complexity in an unbounded and controllable manner, while preserving chirp structure of AFDM.  
Our simulation results demonstrate that the proposed phase function design enhances the PLS performance of AFDM by several orders of magnitude compared with the conventional AFDM in terms of brute-force demodulation complexity.

\end{abstract}

\section{Introduction} 

With the emergence of new use cases and application areas, e.g., high-speed railways, unmanned aerial vehicles (UAVs), and autonomous driving, support for high-mobility communication and sensing has become an essential requirement for next-generation networks \cite{6G1,6G2}. 
In these scenarios, the high-mobility mulitpath wireless channel environment leads to significant delay and Doppler spread.
It is well-known that orthogonal frequency division multiplexing (OFDM) in current standards loses its benefits under such channel conditions \cite{OFDMICI}. 
This motivates the need for new waveforms that are robust to doubly-selective channels while offering sensing capabilities.

With this motivation, new waveforms based on signaling in domains alternative to the conventional time-frequency domain have emerged. Most notably, these include orthogonal time frequency space (OTFS) as a delay-Doppler domain modulation technique, \cite{OTFS}, and affine frequency division multiplexing (AFDM) as an affine domain modulation technique \cite{AFDM}.
Data transmission in these domains enables the exploitation of the full diversity inherent in doubly selective channels, thereby enhancing robustness to channel time variations for reliable communication. 
Moreover, in these domains, delay and Doppler shifts associated with different propagation paths are represented in a structured and separable manner, which is well suited for sensing \cite{ISACdd}. This makes OTFS and AFDM attractive in the context of integrated sensing and communication (ISAC), i.e., a new feature of the sixth generation wireless networks \cite{isac, Isacotfs}. 
However, ISAC introduces physical layer security (PLS) threats, as eavesdroppers may exploit sensing information to estimate the channel state information (CSI) at legitimate users' locations. This highlights the importance of taking into account PLS properties of the new waveforms. 
AFDM, in particular, has recently attracted attention as a chirp-based modulation technique with promising properties for secure communications \cite{brute-force,deltac2,permutation,c_2hopping,c2change,parameter01}. 
The PLS superiority of AFDM over OTFS has been reported in \cite{brute-force}. Despite existing research on the PLS aspects of AFDM \cite{brute-force,deltac2,permutation,c_2hopping,c2change,parameter01}, this line of research remains at an early stage of development, with several important gaps still to be addressed. Hence, the focus of this paper is to further improve the PLS performance of AFDM.

AFDM has two adjustable chirp parameters, namely $c_1$ and $c_2$. The parameter $c_1$ is adjusted to deal with the Doppler effect while $c_2$ offers greater design flexibility \cite{AFDM}. In the context of PLS, $c_1$ and $c_2$ may be kept secret to prevent eavesdroppers from correctly demodulating the signal \cite{AFDMsecure}. However, since $c_1$ is constrained by the maximum Doppler shift of the channel, this paper mainly focuses on $c_2$.  
The intrinsic brute-force complexity of AFDM has been shown to scale quadratically with the number of subcarriers \cite{brute-force}. To further enhance the PLS of AFDM, existing works have exploited $c_2$ in different ways. In \cite{deltac2}, the chirp parameters, including $c_2$ are varied over time across different AFDM symbols through parameter hopping. Although this increases the uncertainty over multiple symbols, the brute-force complexity associated with each individual symbol remains unchanged. In \cite{permutation}, permutations are applied to the $c_2$-related chirp sequence to achieve quantum-resilient PLS. Since this approach does not change $c_2$ itself, it is compatible with $c_2$ variation for further security enhancement. In \cite{c_2hopping}, the $c_2$ values are selected from a pre-designed codebook, which increases the parameter-search space at the eavesdropper. However, the scheme relies on pseudo-random codebook selection rather than a systematic design of the $c_2$-associated phase structure. Therefore, the complexity is not explicitly controlled through the waveform design. In \cite{c2change}, a secure AFDM scheme is proposed in which a set of $c_2$ values is derived from the reciprocal channel between legitimate users. This scheme assumes that the eavesdropper experiences a different channel from that of the legitimate users and therefore cannot obtain the same $c_2$ set for demodulation. While this method redesigns the $c_2$ structure and can increase the demodulation complexity in a more explicit manner, its security is inherently channel-dependent. If the eavesdropper obtains reliable CSI of the legitimate link, the corresponding security advantage may be significantly weakened. This limitation becomes particularly relevant in ISAC scenarios, where sensing functionalities may make channel-related information more accessible to potential eavesdroppers.

Therefore, in this paper, we propose a $c_2$-associated phase function design to improve the physical layer security of AFDM by enhancing its robustness against brute-force demodulation. The proposed method is channel-independent and can render the brute-force search complexity theoretically unbounded, even when the eavesdropper has perfect CSI. In particular, we show that the brute-force robustness of the waveform is governed by the first derivative of the phase function with respect to $c_2$. Motivated by this insight, we introduce a family of phase functions that preserves the chirp structure of AFDM while increasing the brute-force demodulation complexity in a controllable manner. Since the proposed method focuses on parameter design, it can also be integrated with chirp-parameter hopping \cite{deltac2} or chirp permutation \cite{permutation} to further increase the search complexity and enhance security. Our simulation results confirm that the proposed design significantly increases the brute-force demodulation complexity by multiple orders of magnitude.

The rest of the paper is organized as follows. Section II reviews the AFDM signal model and introduces the eavesdropping scenario. The theoretical robustness analysis and the proposed phase design are presented in Section III. Simulation results are presented in Section IV. Finally, Section V concludes the paper.

\subsubsection*{Notations}
Scalars, vectors, and matrices are denoted by italic letters, bold lower-case letters, and bold upper-case letters, respectively. Superscripts ${(\cdot)^{\rm T}}$ and ${(\cdot)^{\rm H}}$ denote transpose and Hermitian transpose, respectively. $\mathbb{R}$ and $\mathbb{C}$ denote the sets of real and complex numbers, respectively. $\mathbf{I}_N$ denotes the $N \times N$ identity matrix, and $\mathrm{diag}(\cdot)$ denotes a diagonal matrix formed from its argument. For a vector $\mathbf{s}$, $s[n]$ denotes its $n$-th element. The notation $[\cdot]_N$ denotes the modulo-$N$ operation. Finally, $j=\sqrt{-1}$ denotes the imaginary unit.

\section{AFDM Principles and Security Aspects}

This section presents the AFDM signal model and the transmission scenario for security analysis in this study. AFDM employs the inverse discrete affine Fourier transform (IDAFT) to map quadrature amplitude modulated (QAM) data symbols from the affine domain to the time domain~\cite{AFDM}. AFDM deploys chirp subcarriers as orthogonal basis functions for data transmission. 
Considering $N$ chirp subcarriers, the modulated signal in time domain can be expressed as
\begin{equation}
s[n] = \sum_{m=0}^{N-1} x_m q_m[n], 
\quad n = 0, \ldots, N-1,
\label{eq:afdm_signal}
\end{equation}
where $x_m$ denotes the $m^{\rm th}$ transmit data symbol and $q_m[n]$ is the transmitter basis function corresponding to the $m^{\rm th}$ chirp subcarrier, i. e.,
\begin{equation}
q_m[n] = \frac{1}{\sqrt{N}}
e^{j2\pi\left(c_1 n^2 + c_2 m^2 + \frac{nm}{N}\right)}.
\end{equation}
The parameters $c_1,c_2\in\mathbb{R}$ are AFDM chirp parameters. 
By stacking the signal and data samples into vectors $\mathbf{s}=[s[0],\ldots,s[N-1]]^{\rm T}$ and $\mathbf{x}=[x_0,\ldots,x_{N-1}]^{\rm T}$, respectively, \eqref{eq:afdm_signal} can be written in matrix form as $\mathbf{s} = \mathbf{Q}\mathbf{x}$,
where
\begin{equation}
\mathbf{Q}
=
\mathbf{\Lambda}_{c_1}^{\mathrm H}
\mathbf{F}_N^{\mathrm H}
\mathbf{\Lambda}_{c_2}^{\mathrm H}
,
\label{eq:afdm_mod_matrix}
\end{equation}
is the unitary AFDM modulation matrix, with
$
\mathbf{\Lambda}_{c}
=
\operatorname{diag}
\left(
[1,
e^{-j2\pi c},
\ldots,
e^{-j2\pi c\,(N-1)^2}]^{\rm T}
\right),
$ $c \in \{c_1,c_2\}$, 
and $\mathbf{F}_N$ denoting the normalized $N$-point discrete Fourier transform (DFT) matrix with the elements $[\mathbf{F}_N]_{p,q}=\frac{1}{\sqrt{N}}e^{\frac{-j2\pi pq}{N}}$ for $p,q=0,\ldots,N-1$.

To cope with the delay spread of the multipath channel and avoid inter symbol interference, a chirp-periodic prefix (CPP) is appended at the beginning of each AFDM symbol. If the maximum delay shift of the channel is $l_{\max}$, the CPP length $L$ should be larger than or equal to $l_{\max}-1$. Therefore, the CPP is defined as
\begin{equation}
s[n] = s[N+n] e^{-j2\pi c_1 \left(N^2 + 2Nn\right)}, \quad n=-L,\ldots,-1.
\end{equation}

After appending CPP, the transmit signal propagates through a doubly dispersive channel. At the receiver side, the received discrete-time signal after CPP removal can be expressed as
\begin{equation}
\mathbf{r} = \mathbf{H}\mathbf{s} + \mathbf{w},
\label{eq_receive_signal}
\end{equation}
where $\mathbf{w}$ denotes the complex additive white Gaussian noise vector with  the variance of $\sigma^2$, i.e., $\mathbf{w}\sim \mathcal{CN}(\mathbf{0},\sigma^2 \mathbf{I}_N)$. The channel matrix $\mathbf{H}$ is represented as
\begin{equation}
\mathbf{H} = \sum_{p=0}^{P-1} 
h_p \, 
\boldsymbol{\Gamma}_{\mathrm{CPP},p} \, 
\boldsymbol{\Pi}^{l_p} \, 
\boldsymbol{\Delta}^{\nu_p},
\end{equation}
where $h_p$, $l_p\in [0,l_{P-1}]$ and $\nu_p \in [-\nu_{\max}, \nu_{\max}]$ denote the complex channel coefficient, delay and Doppler shifts for path $p$, respectively. $P$ is the total number of paths, $l_{\max}=l_{P-1}$ and $\nu_{\max}$ is the maximum Doppler shift of the channel. $\boldsymbol{\Pi}$ is the forward cyclic-shift matrix and $
\boldsymbol{\Delta} =
\operatorname{diag}
\left([1,
e^{-j \frac{2\pi}{N}},
...,
e^{-j  \frac{2\pi}{N}(N-1)}]^{\rm T}
\right)
$ is the unit Doppler shift matrix. 
The matrix $\boldsymbol{\Gamma}_{\mathrm{CPP},p}$ models the effect of the CPP which is given by
\begin{equation}
\boldsymbol{\Gamma}_{\mathrm{CPP},p}
=
\operatorname{diag}\!\left(
\begin{cases}
e^{-j2\pi c_1 (N^2 - 2N(l_p-n))}, & n<l_p, \\
1, & n\ge l_p.
\end{cases}
\right)
\label{gamacpp}
\end{equation}

The received signal in (\ref{eq_receive_signal}) is demodulated and transformed to affine domain by applying the DAFT operation as
\begin{equation}
\begin{aligned}
\mathbf{y}
&= \mathbf{Q}^{\rm H}\mathbf{r} \\
&= \mathbf{\Lambda}_{c_{1}}
\mathbf{F}_N
\mathbf{\Lambda}_{c_{2}}\mathbf{r} \\
&= \mathbf{\Lambda}_{c_{1}}
\mathbf{F}_N
\mathbf{\Lambda}_{c_{2}}
\mathbf{H}
\mathbf{\Lambda}_{c_{1}}^{\rm H}
\mathbf{F}_N^{\rm H}
\mathbf{\Lambda}_{c_{2}}^{\rm H}
\mathbf{x}
+ \tilde{\mathbf{w}} \\
&= \mathbf{H}_{\rm eff}\mathbf{x} + \tilde{\mathbf{w}},
\end{aligned}
\end{equation}
where
$
\mathbf{H}_{\rm eff}
=
\mathbf{\Lambda}_{c_{1}}
\mathbf{F}_N
\mathbf{\Lambda}_{c_{2}}
\mathbf{H}
\mathbf{\Lambda}_{c_{1}}^{\rm H}
\mathbf{F}_N^{\rm H}
\mathbf{\Lambda}_{c_{2}}^{\rm H},
$
and $\tilde{\mathbf{w}} = \mathbf{Q}^{\rm H}\mathbf{w}.$ Since $\mathbf{Q}$ is unitary, $\tilde{\mathbf{w}}$ has the same covariance as $\mathbf{w}$.

For equalization, we consider the minimum mean square error (MMSE) equalizer, given by
\begin{equation}
\hat{\mathbf{x}} = \mathbf{G}_{\rm MMSE}\mathbf{y},
\end{equation}
where
\begin{equation}
\mathbf{G}_{\rm MMSE}
=
\left(
\mathbf{H}_{\rm eff}^{\rm H}\mathbf{H}_{\rm eff}
+
\sigma_w^2 \mathbf{I}_{N}
\right)^{-1}
\mathbf{H}_{\rm eff}^{\rm H},
\end{equation}
with $\sigma_w^2$ denoting the noise variance.

 We consider a scenario in which a malicious eavesdropper attempts to intercept the communication between the base station and the legitimate user equipment. We assume a worst-case scenario in which the eavesdropper is perfectly synchronized and has perfect knowledge of the channel matrix $\mathbf{H}$ as well as the transmitted pilot sequence $\mathbf{x}$. 
 The only unknown parameters are the chirp parameters $c_1$ and $c_2$, which determine the demodulation matrix $\mathbf{Q}$. Although this assumption may appear rather strong, it is commonly adopted in PLS studies to characterize a worst-case eavesdropping scenario \cite{perfectcsi,deltac2}. Moreover, publicly known pilot signals make it reasonable to assume that the pilot sequence is available to the eavesdropper in practical wireless systems. The eavesdropper then attempts to recover the unknown chirp parameters through an exhaustive search procedure.  Accordingly, to mitigate the vulnerability of the legitimate user to brute-force demodulation, in the next section, we propose a PLS technique that increases the complexity of brute-force search to a computationally prohibitive level.
\section{Proposed Security-Oriented Phase Design}

In this section, we propose a phase design technique that improves the robustness of AFDM against brute-force demodulation.
In AFDM, $c_1$ cannot be significantly modified, as it is determined based on the maximum Doppler shift of the channel with $c_1=({2\nu_{\max}+1})/{2N}$ \cite{AFDM}. 
The impact of $c_1$ estimation errors on brute-force demodulation by an eavesdropper has already been analyzed in \cite{brute-force}. By contrast, $c_2$ offers much greater design flexibility as it can vary across chirp subcarriers or AFDM symbols without compromising orthogonality \cite{S-AFDM, subcarriers}.
As shown in \cite{deltac2}, a mismatch in $c_2$ alone is sufficient to prevent an eavesdropper from correctly demodulating the signal. Therefore, $c_2$ is a suitable parameter for security-oriented phase design. Motivated by this, we focus on the design of $c_2$ while assuming that $c_1$ is known or accurately estimated by the eavesdropper. In the following subsections, we first derive a design criterion for $c_2$ in the form of a generic phase function to enhance robustness against brute-force demodulation. Then, we introduce a family of phase functions with adjustable parameters that allow us to increase the brute-force search complexity without a bound.

\subsection{Phase Function Design Principle} 
Substituting the term $c_2 m^2$ with a general function $f(c_2,m)$ in (\ref{eq:afdm_signal}), the AFDM transmit signal can be expressed as
\begin{equation}
s[n] =
\frac{1}{\sqrt{N}}
\sum_{m=0}^{N-1}
x_m
e^{
j2\pi
\left(
c_1 n^2 + \frac{mn}{N} + f(c_2,m)
\right)
}.
\label{eq_new_modulation}
\end{equation}

Considering $\widehat{c}_2$ as the estimate of $c_2$ at the eavesdropper, the  robustness of the system against exhaustive search depends on the mismatch interval denoted by $\Delta_{c_2}$. If $|\widehat{c}_2-c_2| \geq |\Delta_{c_2}|$, the eavesdropper is unable to correctly demodulate the signal. In conventional AFDM, where $f(c_2,m)=c_2m^2$, the periodicity of $c_2$ is 1 and the interval $(0,1]$ represents its effective search range \cite{parameter01}. Hence, the exhaustive-search complexity is proportional to $\frac{1}{\Delta_{c_2}}$ and a smaller mismatch interval leads to improved robustness against brute-force demodulation \cite{brute-force}. Based on this observation, we propose an AFDM phase design criterion, in Proposition~\ref{prop1}, which allows us to tune $\Delta_{c_2}$ at any arbitrarily small value.

\begin{proposition}
For the AFDM modulation with a generalized phase function, $f(c_2,m)$, for any $c_2 \! \in \! \mathbb{R}$ and $m\!=\!0,1,..,N-1$, the admissible mismatch interval $\Delta_{c_2}$ for an eavesdropper to properly demodulate the signal scales as
\begin{equation}
|\Delta_{c_2}|
\propto
\frac{1}{
\left|
\frac{\partial f(c_2,m)}{\partial c_2}
\right|
}.
\label{eq_pro_results}
\end{equation}
\label{prop1}
\end{proposition}

\begin{proof}
Assume that the eavesdropper uses a mismatched chirp parameter
\[
\widehat{c}_2 = c_2 + \Delta_{c_2}
\]
for demodulation. The demodulated symbol $\widehat{x}_k$, for $k=0,1,\dots,N-1$, is given by
\begin{equation}
\widehat x_k =
\frac{1}{\sqrt{N}}
\sum_{n=0}^{N-1}
s[n]
e^{
-j2\pi
\left(
c_1 n^2 + \frac{kn}{N} + f(\widehat{c}_2,k)
\right)
}.
\label{eq_xhat}
\end{equation}
By substituting (\ref{eq_new_modulation}) into \eqref{eq_xhat}, we obtain
\begin{equation}
\begin{aligned}\label{eqn:new_demod}
\widehat x_k 
&=
\frac{1}{N}
\sum_{n=0}^{N-1}
\sum_{m=0}^{N-1}
x_m
e^{
j2\pi
\left(
f(c_2,m)
- f(\widehat{c}_2,k)
+
\frac{(m-k)n}{N}
\right)} \\
&=x_k
e^{
j2\pi
\left(
f(c_2,k)-f(\widehat{c}_2,k)
\right)
},
\end{aligned}
\end{equation}
where the second line in \eqref{eqn:new_demod} is obtained considering
$
\sum_{n=0}^{N-1}
e^{
j2\pi \frac{(m-k)n}{N}
}=N\delta[m-k]
$.

Using the first-order Taylor expansion of $f(\widehat{c}_2,k)=f(c_2+\Delta_{c_2},k)$ in the vicinity of  $c_2$, we have
\begin{equation}\label{eqn:Taylor_f}
\begin{aligned}
f(\widehat{c}_2,k) = f(c_2+\Delta_{c_2},k)\approx
f(c_2,k)
+
\Delta_{c_2}
\frac{\partial f(c_2,k)}{\partial c_2}
.
\end{aligned}
\end{equation}
Substituting \eqref{eqn:Taylor_f} into \eqref{eqn:new_demod}, we have
\begin{equation}\label{eqn:xhat_Taylor}
\widehat x_k
=
x_k
e^{
-j2\pi
\left(
\Delta_{c_2}
\frac{\partial f(c_2,k)}{\partial c_2}
\right)
}.
\end{equation}
For sufficiently small values of $\Delta_{c_2}$, \eqref{eqn:xhat_Taylor} can be approximated as
\begin{equation}
\widehat x_k
\approx
x_k
\left[
1
-j2\pi
\left(
\Delta_{c_2}
\frac{\partial f(c_2,k)}{\partial c_2}
\right)
\right].
\end{equation}
To ensure correct demodulation, we require the residual phase distortion to remain below an arbitrarily small threshold $\varepsilon$, i.e.,
\begin{equation}
2\pi
\left|
\Delta_{c_2}
\frac{\partial f(c_2,k)}{\partial c_2}
\right|
\le
|\varepsilon|,
\end{equation}
where $|\varepsilon| \ll 1$ denotes the maximum tolerable phase error.
Therefore,
\begin{equation}
|\Delta_{c_2}|
\le
\frac{|\varepsilon|}
{2\pi
\left|
\frac{\partial f(c_2,k)}{\partial c_2}
\right|
},
\label{eq_results_FD_SD}
\end{equation}
which proves \eqref{eq_pro_results}.
\end{proof}

For conventional AFDM, where $f(c_2,m)=c_2m^2$, it then follows from Proposition~\ref{prop1} that $\Delta_{c_2}\propto \frac{1}{m^2}$. 
Since Proposition~\ref{prop1} holds for any $m=0,1,..,N-1$, the overall mismatch interval is determined as the smallest achievable $\Delta_{c_2}$ value, which yields
\begin{equation}
    \Delta_{c_2}\propto \frac{1}{(N-1)^2}. 
\end{equation}
As the exhaustive-search complexity is proportional to $\frac{1}{\Delta_{c_2}}$, the brute-force demodulation complexity of conventional AFDM scales as $\mathcal{O}(N^2)$, which is consistent with the result reported in \cite{brute-force}. More importantly, Proposition~\ref{prop1} shows that the mismatch interval is governed by the phase function. As shown in the following subsection, by properly designing the phase function, the mismatch interval can be reduced without bound, thereby enhancing the robustness of AFDM against brute-force demodulation.

\subsection{Proposed Phase Function}\label{subsec:Proposed_Phase_Function}

Based on the above analysis and the result in Proposition~\ref{prop1}, a desirable phase function should ensure that $\left|\frac{\partial f(c_2,m)}{\partial c_2}\right|$ increases rapidly with the subcarrier index $m$. Consequently, the mismatch interval is reduced and the brute-force search complexity is increased. To this end, we propose the following family of phase functions that satisfy the above condition.
\begin{equation}
f(c_2,m)=\kappa m^a \cos\!\left(\pi c_2 m^b\right),
\label{example1}
\end{equation}
where $(\kappa,a,b)\in\mathbb{R}^3$ and $b\ge 0$. Consequently, $\mathbf{\Lambda}_{c_2}^{\mathrm H}$ in \eqref{eq:afdm_mod_matrix} becomes 
\begin{equation}
    \mathbf{\Lambda}_{c_2}^{\mathrm H}
=
\operatorname{diag}
\left(
1,
e^{-j2\pi \kappa \cos(\pi c_2)},
\ldots,
e^{-j2\pi \kappa (N-1)^a \cos(\pi c_2 (N-1)^b)}
\right).
\end{equation}

The first derivative of \eqref{example1} with respect to $c_2$ is
\begin{equation}
\left|\frac{\partial f(c_2,m)}{\partial c_2}\right|
=
\left|\kappa\pi m^{a+b}\sin\!\left(\pi c_2 m^b\right)\right|.
\label{eq_proposed_derivative}
\end{equation}
For operating points such that
\begin{equation}
\left|\sin\!\left(\pi c_2 m^b\right)\right|
\neq 0,
\label{eq_nonzero_sine}
\end{equation}
we obtain $\Delta_{c_2}\propto \frac{1}{m^{a+b}}$ for any $m$ value. Consequently, the system-wide mismatch interval for $c_2$ scales as $\Delta_{c_2}=\mathcal{O}\!\left(\frac{1}{N^{a+b}}\right)$, and the corresponding brute-force complexity for $c_2$ becomes $\mathcal{O}(N^{a+b})$.
In the following, since $|\cos(\pi c_2 m^b)| \in [0,1]$, we fix $a=2$, such that the proposed phase function preserves the quadratic chirp structure of conventional AFDM.
In this case, the complexity becomes 
\begin{equation}
    \mathcal{O}(N^{2+b}).
    \label{eq_scale}
\end{equation}
This result indicates that the brute-force complexity can be increased by a factor of approximately $N^b$ compared with conventional AFDM. Therefore, by properly choosing $b>0$, the proposed phase function can significantly enhance the robustness against brute-force demodulation.

It is worth noting that the above scaling relies on the condition in \eqref{eq_nonzero_sine}. In particular, operating points at which $\sin\!\left(\pi c_2 m^b\right)=0$ should be avoided. For the proposed phase function, $c_2$ can take values within the range $(0,1]$ due to its periodicity of 1. In particular, $c_2=1$ should be avoided. This observation will also be confirmed in the Section~\ref{sec:Simulation}.

In \eqref{example1}, $\kappa$ can be chosen as an irrational number. It can be observed that $\kappa$ introduces an additional degree of freedom, which can further enlarge the exhaustive-search space together with $c_2$. The sensitivity of $\kappa$ can be analyzed similarly using Proposition~\ref{prop1}, by replacing the derivative with $\left|\frac{\partial f(\kappa,m)}{\partial \kappa}\right|$. From \eqref{example1}, we have
\begin{equation}
\left|\frac{\partial f(\kappa,m)}{\partial \kappa}\right|
=
\left|m^a \cos\!\left(\pi c_2 m^b\right)\right|.
\end{equation}
Moreover, there does not exist a $c_2\in(0,1]$ such that $\left|m^2 \cos\!\left(\pi c_2 m^b\right)\right|=0$ for all $m$. Therefore, the mismatch interval for $\kappa$, denoted by $\Delta_{\kappa}$, scales as $\mathcal{O}(\frac{1}{N^a})$, and the corresponding exhaustive-search complexity with respect to $\kappa$ scales as $\mathcal{O}(N^a)$. When $a=2$, the complexity associated with $\kappa$ is of the same order as that associated with $c_2$ in conventional AFDM. This makes it possible to perform chirp hopping on $\kappa$ from one AFDM symbol to the next, since the chirp-hopping technique is originally designed based on the $\Delta_{c_2}$ in conventional AFDM \cite{deltac2}.
When considering both $\kappa$ and $c_2$, the total brute-force demodulation complexity becomes 

\begin{equation}
    \mathcal{O}(N^{b+4}).
    \label{scale}
\end{equation}

Finally, we emphasize that the proposed method preserves the AFDM modulation structure and directly redesigns the phase function associated with $c_2$. Therefore, it remains compatible with other AFDM-based security mechanisms, such as chirp-parameter hopping and chirp permutation.
\section{Simulation Results}
\label{sec:Simulation}
In this section, we numerically evaluate the performance of brute-force demodulation by an eavesdropper when the proposed phase function in \eqref{example1} is adopted for AFDM.
We consider a four-tap LTV channel with the normalized tap powers $\boldsymbol{\rho}=[0.1941,0.4056,0.2388,0.1615]^{\rm T}$ and the corresponding delays of $l_p=p$, $p=0,\dots,3$. 
The normalized Doppler shifts to Doppler spacing are set to $\nu=[0,-0.3,0.8,3]$, with $\nu_{\max}=3$. 
An MMSE equalizer is employed at the receiver. 
We set $N=64$ and $c_1=({2\nu_{\max}+1})/{2N}$. In addition, we set $c_2=0.2$ and $\kappa=\sqrt{2}-1$ unless otherwise stated. The parameter $a$ is set to $2$ to keep the design close to the chirp structure of conventional AFDM. The phase function then becomes $(\sqrt{2}-1)m^2\cos(0.2 \pi m^b)$. All results are obtained using QPSK modulation over $10^5$ Monte Carlo realizations.

In Fig.~\ref{Ber_vs_delta}, we present the BER performance of the proposed phase designs for two different values of $b$, together with conventional AFDM, at a fixed SNR of $25\,\mathrm{dB}$. For conventional AFDM, we set the phase function as $(\sqrt{2}-1)m^2$. The BER is plotted as a function of $\Delta_{c_2}$ for the different phase functions. By taking $10^{-3}$ as the acceptable BER threshold at $25\,\mathrm{dB}$, we observe that the mismatch interval of conventional AFDM is approximately $4\times10^{-5}$, which is consistent with the existing literature \cite{brute-force}. For $b=1$, the mismatch interval is reduced to around $1.7\times10^{-7}$, corresponding to about two orders of magnitude increase in the brute-force search complexity. For $b=10$, the mismatch curve exhibits an impulse-like shape, suggesting an extremely small mismatch interval. In fact, the mismatch interval can be made arbitrarily small as $b$ increases.

\begin{figure}
    \centering
    \includegraphics[width=\columnwidth]{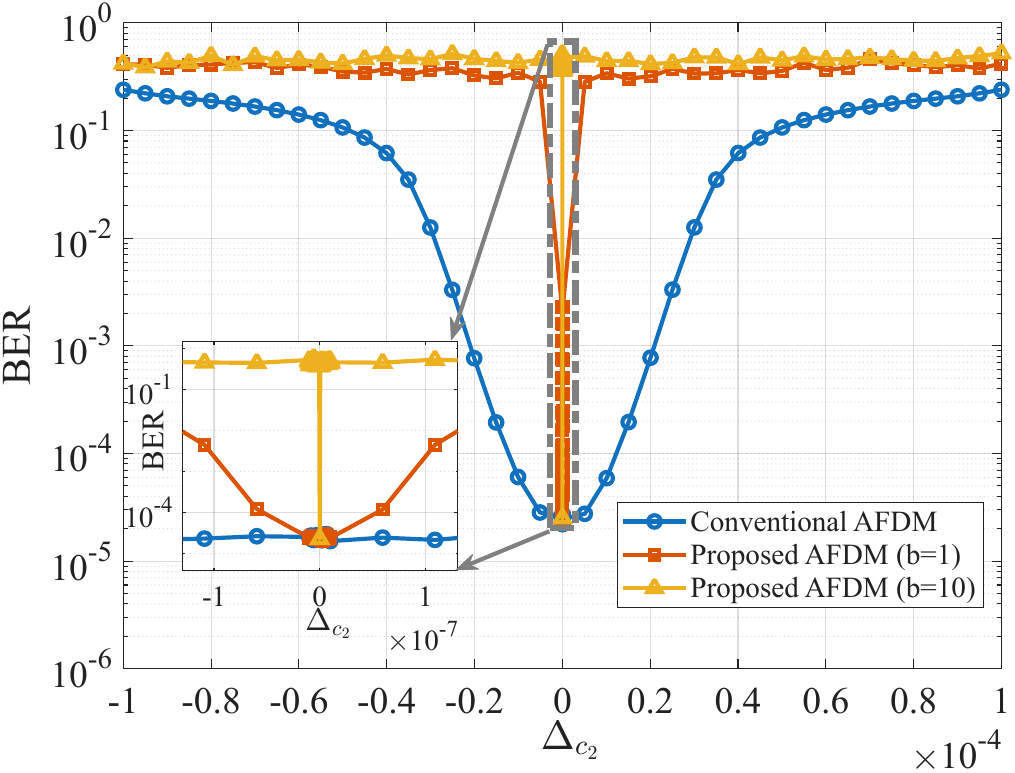}
    \caption{BER versus $\Delta_{c_2}$ for the proposed phase design with different values of $b$, compared with the conventional AFDM.}
    \label{Ber_vs_delta}
\end{figure}
To further examine the mismatch interval, we consider the case where $\Delta_{c_2}=10^{-5}$ for both $b=10$ and conventional AFDM. Iin Fig.~\ref{fig:BERVSSNR}, we plot the BER versus SNR in comparison with the perfectly matched case, i.e., $\Delta_{c_2}=0$. 
As can be seen from Fig.~\ref{Ber_vs_delta}, conventional AFDM achieves a very low BER at $\Delta_{c_2}=10^{-5}$. Therefore, $\Delta_{c_2}=10^{-5}$ is adopted to investigate whether the eavesdropper can still demodulate the signal when $b=10$. The results show that, when $\Delta_{c_2}=0$ (corresponding to perfect estimation, i. e., $\widehat{c}_2=c_2$), AFDM with the conventional phase function and our proposed phase function achieve the same performance. This confirms that our proposed design does not degrade the communication performance of conventional AFDM. 
In contrast, for $\Delta_{c_2}=10^{-5}$, the BER curve of our proposed design with $b=10$ remains at a very high error floor, while conventional AFDM exhibits about the same performance as that of $\Delta_{c_2}=0$. 
This confirms that our proposed design with $b=10$ is far more sensitive than the conventional AFDM to parameter mismatch and therefore, has a significantly smaller mismatch interval.
\begin{figure}
    \centering
    \includegraphics[width=\columnwidth]{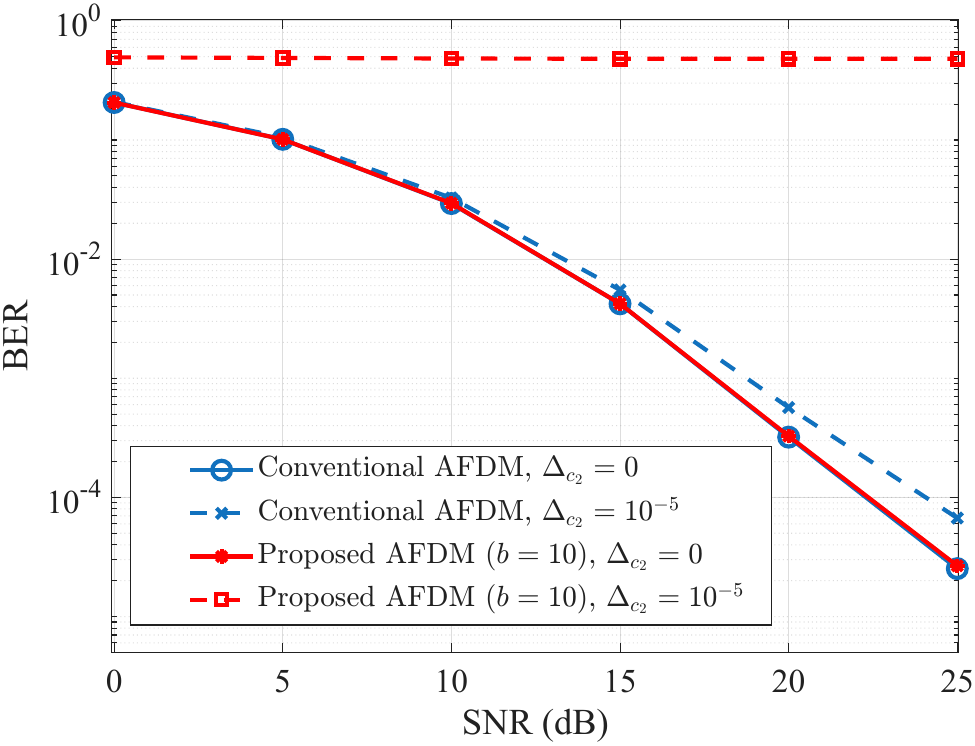}
    \caption{BER versus SNR for conventional AFDM and the proposed AFDM with $b=10$ under the perfectly matched and the mismatched case. 
    }
    \label{fig:BERVSSNR}
\end{figure}

\begin{figure}
    \centering
    \includegraphics[width=\columnwidth]{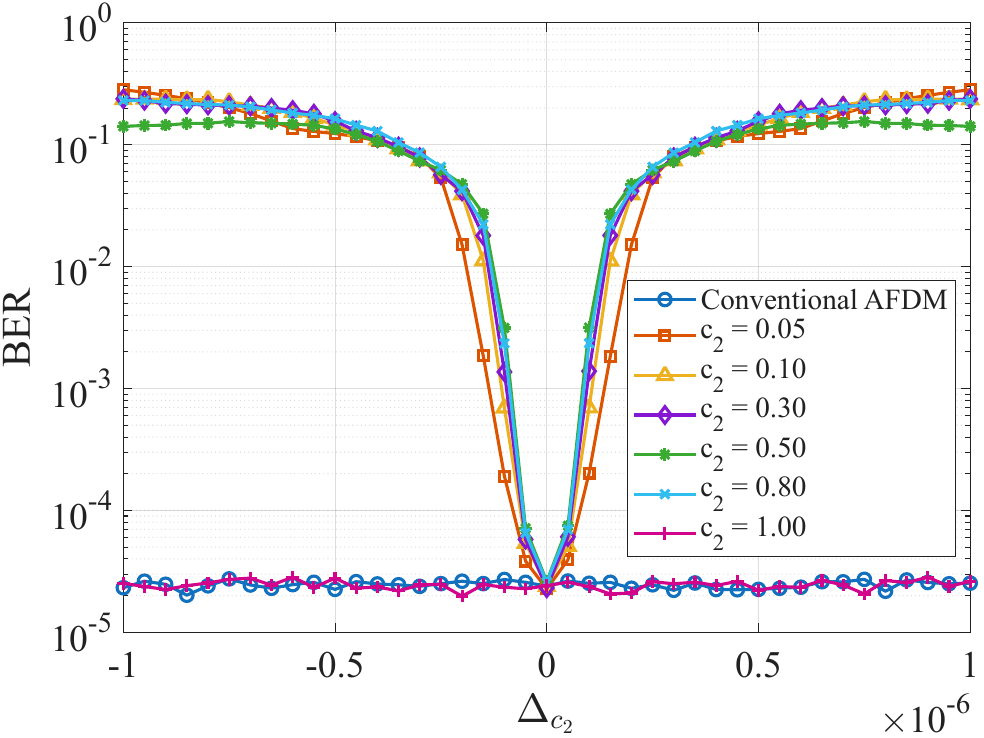}
    \caption{BER versus the parameter mismatch $\Delta_{c_2}$ for different values of $c_2$, compared with the conventional AFDM, with $b=1$.}
    \label{fig:differ_c2}
\end{figure}

Finally, we investigate whether different values of $c_2$ affect the mismatch interval while fixing $b=1$. As shown in Fig.~\ref{fig:differ_c2}, we test different values of $c_2 \in(0,1]$ and observe that, except for the case $c_2=1$, the mismatch interval remains around $1.7\times10^{-7}$. This finding is consistent with the analysis in Section~III. In particular, when $c_2=1$, the condition in \eqref{eq_nonzero_sine} cannot be satisfied. As a result, the mismatch interval increases and PLS performance deteriorates. Therefore, $c_2$ should not be set to $1$. For other values within the range $(0,1)$, the PLS performance remains nearly unchanged. Hence, the admissible range of $c_2$ in the proposed design remains very close to that of conventional AFDM. This allows for parameter hopping on top of the proposed phase design to further improve the robustness of AFDM against brute-force demodulation.

\section{Conclusion}
In this study, we proposed a new class of phase functions associated with the chirp parameter $c_2$ to improve the robustness of AFDM against brute-force demodulation. With the proposed phase functions, the brute-force search complexity can be increased by multiple orders of magnitude in a controllable manner. This allows the legitimate user to flexibly tune the search difficulty faced by the eavesdropper through the phase-function design. Simulation results show that the proposed phase design significantly improves the robustness of the system against eavesdroppers without degrading the communication performance of AFDM. Moreover, since the proposed method directly designs the phase function associated with $c_2$ without changing the underlying AFDM framework, it is flexible enough to be combined with the existing security strategies, such as chirp-parameter hopping and chirp permutation.

\bibliographystyle{IEEEtran}
\bibliography{refs}

\end{document}